\documentclass[11pt]{article}
\usepackage{mathrsfs}
\def\spacingset#1{\def\baselinestretch{#1}\small\normalsize}
\usepackage{subfigure}
\usepackage{graphics}
\usepackage{graphicx}
\usepackage{amsbsy}
\usepackage{amssymb}
\usepackage{amscd}
\usepackage{amsthm}
\usepackage{amsmath}
\usepackage{subfigure}
\usepackage{enumerate}
\usepackage{xcolor}
\usepackage{color}
\usepackage{amsmath,epsfig,fancyhdr, amsthm, amssymb, enumerate}
\usepackage[ruled]{algorithm2e}

\newcommand{\n}{{\cal N} \, }
\newcommand{\ba}{\begin{array}}
\newcommand{\ea}{\end{array}}
\newcommand{\be}{\begin{displaymath}}
\newcommand{\ee}{\end{displaymath}}
\newcommand{\ben}{\begin{equation}}
\newcommand{\een}{\end{equation}}
\newcommand{\bena}{\begin{eqnarray}}
\newcommand{\eena}{\end{eqnarray}}
\newcommand{\beqa}{\begin{eqnarray*}}
\newcommand{\enqa}{\end{eqnarray*}}
\newcommand{\f}{\frac}
\newcommand{\bc}{\begin{center}}
\newcommand{\ec}{\end{center}}
\newcommand{\bi}{\begin{itemize}}
\newcommand{\ei}{\end{itemize}}
\newcommand{\benu}{\begin{enumerate}}
\newcommand{\eenu}{\end{enumerate}}
\newcommand{\bdes}{\begin{description}}
\newcommand{\edes}{\end{description}}
\newcommand{\bt}{\begin{tabular}}
\newcommand{\et}{\end{tabular}}

\newcommand{\Phibf}{\mbox{${\bf \Phi}$}}

\newcommand \nubf{\mbox{\boldmath$\nu$\unboldmath}}

\newcommand \bbf{{\bf b}}
\newcommand \cbf{{\bf c}}
\newcommand \dbf{{\bf d}}

\newcommand \fbf{{\bf f}}
\newcommand \gbf{{\bf g}}

\newcommand \vbf{{\bf v}}
\newcommand \wbf{{\bf w}}
\newcommand \xbf{{\bf x}}
\newcommand \ybf{{\bf y}}
\newcommand \zbf{{\bf z}}

\newcommand \Abf{{\bf A}}
\newcommand \Bbf{{\bf B}}
\newcommand \Cbf{{\bf C}}
\newcommand \Dbf{{\bf D}}

\newcommand \Fbf{{\bf F}}

\newcommand \Ibf{{\bf I}}

\newcommand \Pbf{{\bf P}}

\newcommand \Wbf{{\bf W}}

\newcommand{\Rnum}{{\mathbb R}}

\newcommand{\zebf}{{\bf 0}}

\newcommand{\circlambda}{\mbox{$\Lambda$
             \kern-.85em\raise1.5ex
             \hbox{$\scriptstyle{\circ}$}}\,}

\def\Pr{\mathop{\rm Pr}}

\newtheorem{Theorem}{Theorem}

\begin{document}
%标题部分
\title{\vspace{-0.5cm}
Performance Analysis of $\ell_1$-synthesis with Coherent Frames}
%作者
\author{Yulong Liu, Shidong Li and Tiebin Mi}
%作者
\author{Yulong Liu\thanks{Yulong Liu is with the Institute of Electronics, Chinese Academy of Sciences, Beijing,
$100190$, China. Email: \{yulong3.liu@gmail.com\}.}, Shidong
Li\thanks{Shidong Li is with the Department of Mathematics, San
Francisco State University, San Francisco, CA 94132, USA. Email:
\{shidong@sfsu.edu\}.}, and Tiebin Mi\thanks{Tiebin Mi is with the
School of Information Sciences, Renmin University of China,
$100872$, China. Email: \{mitiebin@gmail.com\}.}}

\date{}

%\begin{document}
\maketitle
 \pssilent \setcounter{page}{1} \thispagestyle{empty}
%\vspace{-0.5cm}
% 开始摘要
%----------------------------------------------------
\begin{abstract}
\spacingset{1.8} Signals with sparse frame representations comprise
a much more realistic model of nature than that with orthonomal
bases. Studies about the signal recovery associated with such
sparsity models have been one of major focuses in compressed
sensing. In such settings, one important and widely used signal
recovery approach is known as $\ell_1$-synthesis (or Basis Pursuit).
We present in this article a more effective performance analysis
(than what are available) of this approach in which the dictionary
$\Dbf$ may be highly, and even perfectly correlated. Under suitable
conditions on the sensing matrix $\Phibf$, an error bound of the
recovered signal $\hat{\fbf}$ (by the $\ell_1$-synthesis method) is
established. Such an error bound is governed by the decaying
property of $\tilde{\Dbf}_{\text{o}}^*\fbf$, where $\fbf$ is the
true signal and $\tilde{\Dbf}_{\text{o}}$ denotes the optimal dual
frame of $\Dbf$ in the sense that
$\|\tilde{\Dbf}_{\text{o}}^*\hat{\fbf}\|_1$ produces the smallest
$\|\tilde{\Dbf}^*\tilde{\fbf}\|_1$ in value among all dual frames
$\tilde{\Dbf}$ of $\Dbf$ and all feasible signals $\tilde{\fbf}$.
This new performance analysis departs from the usual description of
the combo $\Phibf\Dbf$, and places the description on $\Phibf$.
Examples are demonstrated to show that when the usual analysis fails
to explain the working performance of the synthesis approach, the
newly established results do.

%开始关键词

\end{abstract}
%---------------------------------------------------
\spacingset{1.9}
%%%%%%%%%%%%%%%%%%%%%%%%%%%%%%%%%
\section{Introduction}
%%%%%%%%%%%%%%%%%%%%%%%%%%%%%%%%%
Compressed sensing is a new data acquisition theory which allows
that sparse or compressible signals of interest can be recovered
from a small number of linear, non-adaptive, and usually randomized
measurements \cite{Candes2006b, Candes2006c, Donoho2006}. By now,
compressed sensing has attacked  abundant applications in signal and
image processing, see e.g., the two special issues
\cite{IEEESP2008}, \cite{ProcIEEE2010} and references therein.
Formally, one considers the following measurement model:
\begin{equation}
  \label{MeasurementModel} \ybf = \Phibf \fbf + \zbf,
\end{equation}
where $\Phibf$ is an $m \times n$ sensing matrix with $m \ll n$
(indicating some significant undersampling) and $\zbf \in \Rnum^{m}$
is a noise term modeling measurement error. The goal is to
reconstruct the unknown signal $\fbf \in \Rnum^{n}$ based on
available measurements $\ybf \in \Rnum^{m}$.

In standard compressed sensing scenarios, it is usually assumed that
the signal $\fbf$ has a sparse (or nearly sparse) representation in
an orthonormal basis. However, a large number of applications in
signal and image processing point to problems where $\fbf$ is sparse
with respect to an overcomplete dictionary or a frame rather than an
orthonormal basis, see, e.g., \cite{Mallat1993}, \cite{Chen2001},
\cite{Bruckstein2009}, and references therein. Examples include,
e.g., signal modeling in array signal processing (oversampled array
steering matrix), reflected radar and sonar signals (Gabor frames),
and images with curves (curvelets), etc. The flexibility of frames
is the key characteristic that empowers frames to become a natural
and concise signal representation tool. Therefore, it is highly
desirable to extend the compressed sensing methodology to redundant
dictionaries as apposed to orthnormal bases only, see, e.g.,
\cite{Rauhut2008}, \cite{Candes2011}, \cite{Liu2011}. In such sparse
frame representation setting, the signal $\fbf$ is now expressed as
$\fbf=\Dbf\xbf$ where $\Dbf\in \Rnum^{n\times d}$ ($n < d$) is a
matrix of frame\footnote{A set of vectors $\{\dbf_{k}\}_{k\in I}$ in
$\Rnum^{n}$ is a \textit{frame} of $\Rnum^{n}$ if there exist
constants $0<A\leq B<\infty$ such that
\begin{equation*}
  \forall \  \fbf \in \Rnum^{n}, \ \ \ \ A\|\fbf\|_{2}^{2} \leq \sum_{k\in I}|\langle\fbf, \dbf_{k}
  \rangle|^{2} \leq B\|\fbf\|_{2}^{2},
\end{equation*}
where numbers $A$ and $B$ are called lower and upper frame bounds,
respectively. More details about frames can be found in e.g.,
\cite{Christensen2003}, \cite{Han2007}.} vectors (as columns) that
are often rather coherent in applications, and $\xbf\in \Rnum^{d}$
is a sparse coefficient vector.  The linear measurements of $\fbf$
then can be written as
\begin{equation} \label{MeasurementModelofADx}
\ybf = \Phibf \Dbf\xbf + \zbf.
\end{equation}

Since $\xbf$ is assumed sparse, the standard way of recovering
$\fbf$ from (\ref{MeasurementModelofADx}) is known as
$\ell_{1}$-synthesis (or synthesis-based method) \cite{Chen2001},
\cite{Elad2007}, \cite{Candes2011}. From the measurements, one first
finds the sparsest possible coefficient $\xbf$ by solving an
$\ell_{1}$ minimization problem
\begin{equation}
\label{L1synthesis} \hat{\xbf}=\underset {\tilde{\xbf}\in\Rnum^{d}}
{\textrm{argmin}} \Vert\tilde{\xbf}\Vert_{1} \ \ \ \ s.t. \ \
\Vert\ybf-\Phibf\Dbf\tilde{\xbf}\Vert_{2} \leq \epsilon,
\end{equation}
where $\Vert\xbf\Vert_{p} \ (p=1,2)$ denotes the standard
$\ell_{p}$-norm of the vector $\xbf$ and $\epsilon$ is an upper
bound of the noise\footnote{The extension to the Gaussian noise case
is straightforward since with large probability, the Gaussian noise
belongs to bounded sets, see, e.g., lemma $5.1$ in \cite{Cai2009}.}.
Then the solution to $\fbf$ is derived via a synthesis operation,
i.e., $\hat{\fbf}=\Dbf\hat{\xbf}$.

Although empirical studies show that $\ell_{1}$-synthesis often
achieves good recovery results, the theoretical performance of this
method is far from satisfactory. The analytical results in
\cite{Rauhut2008} essentially require that the frame $\Dbf$ has
columns that are extremely uncorrelated such that the compound
matrix $\Phibf\Dbf$ satisfies the requirements imposed by the
traditional compressed sensing assumptions. However, these
requirements are often infeasible when $\Dbf$ is highly coherent.
For example, consider a simple case in which $\Phibf\in
\Rnum^{m\times n}$ is a Gaussian matrix with i.i.d. entries, then
$\Phibf\sim\n(\zebf, \Ibf_{n}\otimes \Ibf_{m})$, where $\otimes$
denotes the Kronecker product and $\Ibf_{m}$ is an identity matrix
of the size $m$.  It is now well known that with very high
probability $\Phibf$ satisfies the restricted isometry property
(RIP) \cite{Candes2005} provided that $m$ is on the order of
$s\log(n/s)$ \cite{Baraniuk2008}, \cite{Candes2006b}. Let us now
examine $\Phibf\Dbf$.   It is not hard to show that
$\Phibf\Dbf\sim\n(\zebf, \Dbf^{*}\Dbf\otimes\Ibf_{m})$, where
$(\cdot)^{*}$ denotes the transpose operation. If $\Dbf$ is unitary,
then $\Phibf\Dbf$ has the same distribution as $\Phibf$, and hence
satisfies the RIP. However, if $\Dbf$ is a coherent frame, then
$\Phibf\Dbf$ may no longer obey the common RIP since the entries of
$\Phibf\Dbf$ are correlated. Meantime, the mutual incoherence
property (MIP) \cite{Donoho2006a} may not apply either, as it is
very hard for $\Phibf\Dbf$ to satisfy the MIP as well when $\Dbf$ is
highly correlated.

The perspective of the results in \cite{Rauhut2008} is that some
sufficient conditions are put on the compound matrix $\Phibf\Dbf$
such that $\xbf$ can be recovered accurately, which leads to a good
estimate of $\fbf$. However, if one is only interested in
reconstructing the signal $\fbf$ and may not care about obtaining a
good recovery of $\xbf$. As pointed out in \cite{Candes2011}, when
the dictionary $\Dbf$ has two identical columns, it seems impossible
to recover a unique sparse coefficient vector $\xbf$ from the
measurements, but we may certainly be able to reconstruct the signal
$\fbf$ accurately. In other words, a good recovery of $\xbf$ may be
unnecessary to guarantee an accurate reconstruction of $\fbf$.

We observe in abundant examples that the $\ell_1$-synthesis method
is also capable of producing fine approximation of $\fbf$ without
recovering accurate coefficient vector $\xbf$.  Known analysis
results such as \cite{Rauhut2008} would then not be able to explain
these fine results by the synthesis approach.

In this article, we present a new performance analysis for the
$\ell_1$-synthesis approach \eqref{L1synthesis} in which the
dictionary $\Dbf$ may be highly - and even perfectly - correlated.
To the best knowledge of the authors, our new results are more
effective than what are known and available.  Our results do not
depend on a good recovery of the coefficients. The basic idea is to
establish the equivalence between the $\ell_1$-synthesis approach
and the optimal-dual-based $\ell_1$-analysis approach recently
proposed in \cite{Liu2011}. Then the recovery error bound for the
latter will naturally lead to that for the former.

This paper is organized as follows. Section \ref{section2}
introduces the family of analysis-based approaches which includes
the standard $\ell_1$-analysis, the general-dual-based
$\ell_1$-analysis, and the optimal-dual-based $\ell_1$-analysis. In
Section \ref{section3}, the equivalence between the
$\ell_1$-synthesis and the optimal-dual-based $\ell_1$-analysis is
established.  The new performance analysis (error bound) for the
$\ell_1$-synthesis is then naturally followed from that of the
optimal-dual based $\ell_1$-analysis approach. Some numerical
experiments are presented in Section \ref{section4} to demonstrate
the effectiveness of the results obtained in Section \ref{section3}.
These examples show that when the usual analysis fails to explain
the working performance of the synthesis approach, our newly
established results do. Conclusion remarks are given in Section
\ref{section5}.

%%%%%%%%%%%%%%%%%%%%%%%%%%%%%%%%%
\section{The Family of Analysis-based Approaches Based on General Dual Frames} \label{section2}
%%%%%%%%%%%%%%%%%%%%%%%%%%%%%%%%%
Alongside the $\ell_1$-synthesis approach, there is a counterpart
that takes an analysis point of view, see e.g., \cite{Elad2005},
\cite{Elad2007}, \cite{Candes2011}. This alternative finds an
estimate of $\fbf$ directly by solving the problem
\begin{equation}\label{StandardL1Analysis}
  \hat{\fbf}=\underset {\tilde{\fbf}\in\Rnum^{n}} {\textrm{argmin}} \Vert\bar{\Dbf}^{*}\tilde{\fbf}\Vert_{1} \ \ \ \  s.t. \ \
  \Vert\ybf-\Phibf\tilde{\fbf}\Vert_{2} \leq \epsilon,
 \end{equation}
where $\bar{\Dbf}$ denotes the canonical dual frame of $\Dbf$, i.e.,
$\bar{\Dbf} = (\Dbf\Dbf^*)^{-1}\Dbf$. Note that if $\Dbf$ is a
Parseval frame, then we have $\bar{\Dbf} = \Dbf$.

It is well known by now \cite{Elad2007} that when $\Dbf$ is a square
and invertible dictionary, the $\ell_{1}$-analysis and
$\ell_{1}$-synthesis approaches are equivalent. However, when $\Dbf$
is an overcomplete frame, the gap between them exists.

A remarkable performance study of the $\ell_{1}$-analysis approach
\eqref{StandardL1Analysis} in the case of Parseval frames
($\bar{\Dbf} = \Dbf$) was given in \cite{Candes2011}. It was shown
that, under suitable conditions on the sensing matrix $\Phibf$, the
solution to \eqref{StandardL1Analysis} is very accurate provided
that $\Dbf^*\fbf$ has rapidly decreasing coefficients. In other
words, when the frame coefficient vector $\Dbf^{*}\fbf$ is
reasonably sparse, $\ell_{1}$-analysis can be the right method to
use.

However, that $\fbf$ is sparse in terms of $\Dbf$ does not imply
$\Dbf^{*}\fbf$ is necessarily sparse. In fact, as the canonical dual
frame expansion in the case of Parseval frames,
$\Dbf^{*}\fbf=\Dbf^{*}\Dbf\xbf$ has the minimum $\ell_{2}$-norm by
the frame property, see, e.g., \cite{Christensen2003} and is usually
fully populated which is also pointed out in \cite{Rauhut2008}. In
other words, the canonical dual frame of $\Dbf$ may be ineffective
in sparsifying $\fbf$ since $\ell_{2}$-norm tends to spread the
coefficients into a large number of small coefficients.

To overcome this difficulty, the standard $\ell_1$-analysis approach
\eqref{StandardL1Analysis} has recently been extended to a more
general case in which the analysis operator can be any dual
frame\footnote{A frame $\{\widetilde{\dbf}_{k}\}_{k\in I}$ is an
alternative dual frame of $\{\dbf_{k}\}_{k\in I}$ if
\begin{equation*}
  \forall \  \fbf \in \Rnum^{n}, \ \ \ \ \fbf = \sum_{k\in I}\langle\fbf, \widetilde{\dbf}_{k}
  \rangle\  \dbf_{k} = \sum_{k\in I}\langle\fbf, {\dbf}_{k}
  \rangle\  \widetilde{\dbf}_{k}.
\end{equation*}} of $\Dbf$ \cite{Liu2011}. This leads to the
following general-dual-based $\ell_1$-analysis approach
\begin{equation}
 \label{generall1analysis}
 \hat{\fbf}=\underset {\ \tilde{\fbf}\in\Rnum^{n}} {\textrm{argmin}}
 \Vert\tilde{\Dbf}^{*}\tilde{\fbf}\Vert_{1} \ \ \ \  s.t. \ \
  \Vert\ybf-\Phibf\tilde{\fbf}\Vert_{2} \leq \epsilon,
\end{equation}
where columns of $\tilde{\Dbf}$ form a general (and any) dual frame
of $\Dbf$. The performance analysis of the general-dual-based
$\ell_1$-analysis approach was also given in \cite{Liu2011}. In
order to introduce the results, we require the concept of $\Dbf$-RIP
\cite{Candes2011}: An $m \times n$ sensing matrix $\Phibf$ is said
to satisfy the restricted isometry property adapted to $\Dbf$
(abbreviated $\Dbf$-RIP) with constant $\delta_{s}\in (0,1)$ if
\begin{equation}\label{DRIP}
(1-\delta_{s})\Vert \vbf \Vert_{2}^{2}\leq \Vert \Phibf\vbf
\Vert_{2}^{2} \leq(1+\delta_{s})\Vert \vbf \Vert_{2}^{2}
\end{equation}
holds for all $\vbf\in \Sigma_{s}$, where $\Sigma_{s}$ is the union
of all subspaces spanned by all subsets of $s$ columns of $\Dbf$.
The validity of the $\Dbf$-RIP was discussed in e.g.,
\cite{Candes2011}, \cite{Krahmer2011}. It was shown in
\cite{Candes2011} that any $m \times n$ matrix $\Phibf$ obeying for
any fixed $\nubf \in \Rnum^{n}$
\begin{equation}
  \label{ConcentrationInequality}
  \Pr\left( \left|\|\Phibf\nubf\|_2^2 - \|\nubf\|_2^2 \right| \geq
  \delta \|\nubf\|_2^2 \right) \leq
  c e^{-\gamma \delta^2 m}, \ \ \ \ \delta \in (0, 1)
\end{equation}
($\gamma$, $c$ are positive constants) will satisfy the $\Dbf$-RIP
with overwhelming probability provided that $m$ is on the order of
$s\log(d/s)$. Many types of random matrices satisfy
\eqref{ConcentrationInequality}, some examples include matrices with
Gaussian, subgaussian, or Bernoulli entries. It has also been shown
in \cite{Krahmer2011} that randomizing the column signs of any
matrix that satisfies the standard RIP results in a matrix which
satisfies the Johnson-Lindenstrauss lemma \cite{Johnson1984}. Such a
matrix would then satisfy the $\Dbf$-RIP via
\eqref{ConcentrationInequality}. Consequently, partial Fourier
matrix (or partial circulant matrix) with randomized column signs
will satisfy the $\Dbf$-RIP since these matrices are known to
satisfy the RIP.

With these preliminaries, we now restate the results in
\cite{Liu2011} as follows.

\begin{Theorem}\label{thm1} \cite{Liu2011}
  Let $\Dbf$ be a general frame of $\Rnum^{n}$ with frame bounds $0<A\leq B<\infty$.  Let $\tilde{\Dbf}$ be an alternative
  dual frame of $\Dbf$ with frame bounds $0<\tilde{A}\leq \tilde{B}<\infty$, and let $\rho=s/b$. Suppose
  \begin{equation}
    \label{SufficientCondition} \left(1-\sqrt{\rho B \tilde{B}}\right)^2 \cdot \delta_{s+a} +
\rho B \tilde{B}\cdot\delta_{b} < 1 - 2\sqrt{\rho B \tilde{B}}
  \end{equation}
holds for some positive integers $a$ and $b$ satisfying $0< b-a\leq
3a$. Then the solution $\hat{\fbf}$ to \eqref{generall1analysis}
satisfies
  \begin{equation}
    \label{ErrorBoundofGDBL1analysis} \Vert \hat{\fbf}-\fbf \Vert_{2} \leq C_{0}\cdot\epsilon +
    C_{1}\cdot\f{\Vert\tilde{\Dbf}^{*}\fbf-(\tilde{\Dbf}^{*}\fbf)_{s}\Vert_{1}}{\sqrt{s}},
  \end{equation}
 where $C_{0}$ and $C_{1}$ are some constants and $(\tilde{\Dbf}^{*}\fbf)_{s}$
 denotes the vector consisting the largest $s$ entries of
 $\tilde{\Dbf}^{*}\fbf$ in magnitude (and setting the other to zero).
\end{Theorem}

Theorem \ref{thm1} shows that if $\Phibf$ satisfies some proper
conditions, e.g., \eqref{SufficientCondition}, then the solution to
\eqref{generall1analysis} is very accurate provided that
$\tilde{\Dbf}^*\fbf$ has rapidly decreasing coefficients. By the
definition of the $\Dbf$-RIP, the condition
\eqref{SufficientCondition} is independent of the coherence of the
dictionary $\Dbf$. For differently chosen $a$ and $b$,
\eqref{SufficientCondition} will give rise to different conditions
on the $\Dbf$-RIP constants $\delta_{s+a}$ and $\delta_{b}$. For
instance, if $\Dbf$ is a Parseval frame and $\tilde{\Dbf}$ is its
canonical dual frame, i.e., $B\tilde{B}=1$, then
\eqref{SufficientCondition} is satisfied whenever
$\delta_{2s}<0.1398$ \cite{Liu2011}.

With the error bound \eqref{ErrorBoundofGDBL1analysis}, we can
easily see the potential superiority of using alternative dual
frames as analysis operators. For clarity, we consider a simple case
in which the noise is free, i.e., $\epsilon = 0$. Then the error
bound \eqref{ErrorBoundofGDBL1analysis} reduces to
  \begin{equation}
    \label{ErrorBoundNoiseFree} \Vert \hat{\fbf}-\fbf \Vert_{2} \leq
    C_{1}\cdot\f{\Vert\tilde{\Dbf}^{*}\fbf-(\tilde{\Dbf}^{*}\fbf)_{s}\Vert_{1}}{\sqrt{s}}.
  \end{equation}
Clearly, the quality of the bound
${\Vert\tilde{\Dbf}^{*}\fbf-(\tilde{\Dbf}^{*}\fbf)_{s}\Vert_{1}}/{\sqrt{s}}$
in \eqref{ErrorBoundNoiseFree} is measured in terms of how effective
$\tilde{\Dbf}^{*}\fbf$ is in spasifying the signal $\fbf$ with
respect to the dictionary $\Dbf$. To explain, suppose that $\fbf$
has a sparse representation in $\Dbf$, i.e., $\fbf = \Dbf\xbf$,
where $\xbf$ is a sparse coefficient vector. As discussed before,
the canonical dual frame expansion of $\fbf$ has the minimum
$\ell_{2}$-norm, i.e., $\|\bar{\Dbf}^*\fbf\|_2 = \underset
{\tilde{\xbf}: \Dbf\tilde{\xbf}=\fbf} {\textrm{min}}
\|\tilde{\xbf}\|_2$, and is ineffective in promoting sparsity in
general. On the other hand, when the analysis operator can be any
dual frame of $\Dbf$, it is not hard to imagine that there should be
some dual frame of $\Dbf$, denoted by $\tilde{\Dbf}_{\mathcal {S}}$,
such that $\tilde{\Dbf}_{\mathcal {S}}^*\fbf = \xbf$. This is due to
the fact that all coefficients of a frame expansion of $\fbf$ in
$\Dbf$ should correspond to some dual frame of $\Dbf$, which really
is the spirit of frame expansions. Generally,
$\tilde{\Dbf}_{\mathcal {S}}$ is much more effective in sparsifying
the signal $\fbf$ than the canonical dual frame does. Therefore, one
may expect a better recovery performance by taking some ``proper''
alternative dual frame as the analysis operator.

The important question then is how to choose some appropriate dual
frame such that the corresponding analysis coefficients are as
sparse as possible. Since the true $\fbf$ is never known before hand
in practice, it seems to impossible to explicitly construct some
proper dual frame $\tilde{\Dbf}$ such that $\tilde{\Dbf}^*\fbf$ is
sparse without additional priori knowledge about the signal $\fbf$.
One approach proposed in \cite{Liu2011} is by the method of
optimal-dual-based $\ell_1$-analysis:
\begin{equation}
 \label{optimaldualbasedL1analysis1}
 \hat{\fbf}=\underset {\Dbf\tilde{\Dbf}^* = \Ibf, \ \tilde{\fbf}\in\Rnum^{n}} {\textrm{argmin}}
 \Vert\tilde{\Dbf}^{*}\tilde{\fbf}\Vert_{1} \ \ \ \  s.t. \ \
  \Vert\ybf-\Phibf\tilde{\fbf}\Vert_{2} \leq \epsilon,
\end{equation}
where the optimization is performed simultaneously over both all
dual frames $\tilde{\Dbf}$ of $\Dbf$ and the feasible signal set.
This seemingly complicated optimization problem can be reformulated
into a simplified form. Note that the class of all dual frames for
$\Dbf$ is given by \cite{Li1995}
\begin{equation} \label{generaldual}
 \tilde{\Dbf}  = (\Dbf\Dbf^*)^{-1}\Dbf +
 \Wbf^*(\Ibf_d-\Dbf^*(\Dbf\Dbf^*)^{-1}\Dbf)= \bar{\Dbf} + \Wbf^*\Pbf,
\end{equation}
where $\Pbf\equiv \Ibf_d-\Dbf^*(\Dbf\Dbf^*)^{-1}\Dbf$ denotes the
orthogonal projection onto the null space of $\Dbf$ and $\Wbf \in
\Rnum^{d \times n}$ is an arbitrary matrix. Plug \eqref{generaldual}
into \eqref{optimaldualbasedL1analysis1}, we obtain
\begin{equation}
 \label{optimaldualbasedL1analysis2}
 (\hat{\fbf}, \ \hat{\gbf})=\underset {\tilde{\fbf}\in\Rnum^{n},\  \gbf\in\Rnum^{d}} {\textrm{argmin}}
 \Vert\bar{\Dbf}^{*}\tilde{\fbf} + \Pbf \gbf\Vert_{1} \ \ \ \  s.t. \ \
  \Vert\ybf-\Phibf\tilde{\fbf}\Vert_{2} \leq \epsilon,
\end{equation}
where we have used the fact that when $\tilde{\fbf}\neq \zebf$,
$\gbf \equiv \Wbf\tilde{\fbf}$ can be any vector in $\Rnum^{d}$ due
to the fact that $\Wbf$ is free. Note that if $\Pbf \gbf \equiv
\zebf$, then \eqref{optimaldualbasedL1analysis2} reduces to the
standard $\ell_1$-analysis approach \eqref{StandardL1Analysis}. In
\cite{Liu2011}, an iterative algorithm based on the split Bregman
iteration \cite{Goldstein2009} was developed to solve the
optimization problem \eqref{optimaldualbasedL1analysis2}
efficiently.

Clearly, the solution to \eqref{optimaldualbasedL1analysis1}
definitely corresponds to that of \eqref{generall1analysis} with
some ``optimal'' dual frame, say $\tilde{\Dbf}_{\text{o}}$ as the
analysis operator. The optimality here is in the sense that
$\|\tilde{\Dbf}_{\text{o}}^*\hat{\fbf}\|_1$ achieves the smallest
$\|\tilde{\Dbf}^*\tilde{\fbf}\|_1$ in value among all dual frames
$\tilde{\Dbf}$ of $\Dbf$ and feasible signals $\tilde{\fbf}$
satisfied the constraint in \eqref{optimaldualbasedL1analysis1}.
Once $\hat{\fbf}$ and $\hat{\gbf}$ are obtained (through solving
\eqref{optimaldualbasedL1analysis2}), it follows from
\eqref{generaldual} that the analysis operator
$\tilde{\Dbf}_{\text{o}}^*$ is given by
\begin{equation}
  \label{optimalduals}
  \tilde{\Dbf}_{\text{o}}^* = \bar{\Dbf}^* + \Pbf\Wbf_{\text{o}},
\end{equation}
with  $\Wbf_{\text{o}}$ satisfying
\begin{equation} \label{EquationofW}
\hat{\gbf} = \Wbf_{\text{o}}\hat{\fbf}.
\end{equation}
Evidently, the optimal dual frame $\tilde{\Dbf}_{\text{o}}$ depends
on the solutions of \eqref{optimaldualbasedL1analysis2}. By
utilizing the fact that $\text{vec}(\Abf\Bbf\Cbf)=(\Cbf^*\otimes
\Abf)\text{vec}(\Bbf)$, the above equation \eqref{EquationofW} is
equivalent to
\begin{equation}
  \label{EquationofWEQ}
  (\hat{\fbf}^*\otimes \Ibf_d) \cdot
  \text{vec}(\Wbf_{\text{o}}) = \hat{\gbf},
\end{equation}
where $\text{vec}(\Wbf_{\text{o}})$ denotes the vectorization of the
matrix $\Wbf_{\text{o}}$ by stacking the columns of
$\Wbf_{\text{o}}$ into a single vector. Evidently, the solution to
\eqref{EquationofWEQ} is non-unique in general since this equation
is highly underdetermined with $n$ equations but $nd$ unknowns. The
class of solutions to \eqref{EquationofWEQ} is given by
\begin{align}\label{LSofW}
 \text{vec}(\Wbf_{\text{o}}) & = (\hat{\fbf}^*\otimes
 \Ibf_d)^\dag \hat{\gbf} + \left(\Ibf_{nd}-(\hat{\fbf}^*\otimes
 \Ibf_d)^\dag (\hat{\fbf}^*\otimes
 \Ibf_d)\right)\wbf \notag \\
 & = (\hat{\fbf}\otimes
 \Ibf_d) \hat{\gbf}/\|\hat{\fbf}\|_2^2 + \left(\Ibf_{nd}-(\hat{\fbf}\hat{\fbf}^*\otimes
 \Ibf_d)\right)\wbf/\|\hat{\fbf}\|_2^2,
\end{align}
where $\Abf^\dag$ denotes the pseudo-inverse of $\Abf$ and $\wbf \in
\Rnum^{nd\times 1}$ is an arbitrary vector. In deriving
\eqref{LSofW}, we have used the two facts that $(\Abf \otimes
\Bbf)^\dag = \Abf^\dag \otimes \Bbf^\dag$ and $(\Abf \otimes
\Bbf)(\Cbf \otimes \Dbf) = \Abf\Cbf \otimes \Bbf\Dbf$, see e.g.,
\cite{Lutkepohl1996}. Let $\wbf = \zebf$, then \eqref{LSofW} reduces
to the least square solution of $\Wbf_{\text{o}}$
\begin{equation}\label{WLS}
 \Wbf_{\text{o}}^{\text{ls}} = (\hat{\fbf}^*\otimes
 \hat{\gbf}) /\|\hat{\fbf}\|_2^2.
\end{equation}
If we choose $\Wbf_{\text{o}}=\Wbf_{\text{o}}^{\text{ls}}$, then
\eqref{optimalduals} becomes
\begin{equation}\label{OptimaldualLS}
\tilde{\Dbf}_{\text{o}}^* = \bar{\Dbf}^* + \Pbf
\Wbf_{\text{o}}^{\text{ls}} = \bar{\Dbf}^* + (\hat{\fbf}^*\otimes
 \Pbf\hat{\gbf}) /\|\hat{\fbf}\|_2^2.
\end{equation}
It is this form \eqref{OptimaldualLS} which will be used to
construct the optimal dual frame in the numerical experiments.

Figure \ref{Fig1} provides a schematic overview of the family of
dual-based $\ell_1$-analysis approaches. For the standard
$\ell_1$-analysis approach \eqref{StandardL1Analysis} which uses the
canonical dual frame of $\Dbf$ as the analysis operator, the
recovered signal $\hat{\fbf}$ has the smallest
$\|\bar{\Dbf}^*\tilde{\fbf}\|_1$ in value among the feasible signal
set. While for the optimal-dual-based $\ell_1$-analysis approach
\eqref{optimaldualbasedL1analysis1}, the optimization is not only
over the feasible signal set but also over all dual frames
$\tilde{\Dbf}$ of $\Dbf$. The recovered signal $\hat{\fbf}$ and
optimal dual frame $\tilde{\Dbf}_{\text{o}}$ (non-unique) produce
the smallest $\|\tilde{\Dbf}^*\tilde{\fbf}\|_1$ in value. When the
signal of interest has a sparse representation in a redundant frame,
one may expect that the optimal dual frame may be much effective in
sparsfying the true signal than the canonical dual frame does. Then
a better recovery performance may be achieved by the
optimal-dual-based $\ell_1$-analysis approach. Indeed, we have seen
that the signal recovery via \eqref{optimaldualbasedL1analysis1} is
much more effective than that of the standard $\ell_1$-analysis
approach \eqref{StandardL1Analysis} which uses the canonical dual
frame as the analysis operator. Moreover, the optimal-dual-based
analysis method provides a new and more effective performance
analysis to the $\ell_1$-synthesis approach.

 \begin{figure}
 \begin{center}
 \begin{tabular}{c}
 \includegraphics[height=5.5cm]
 {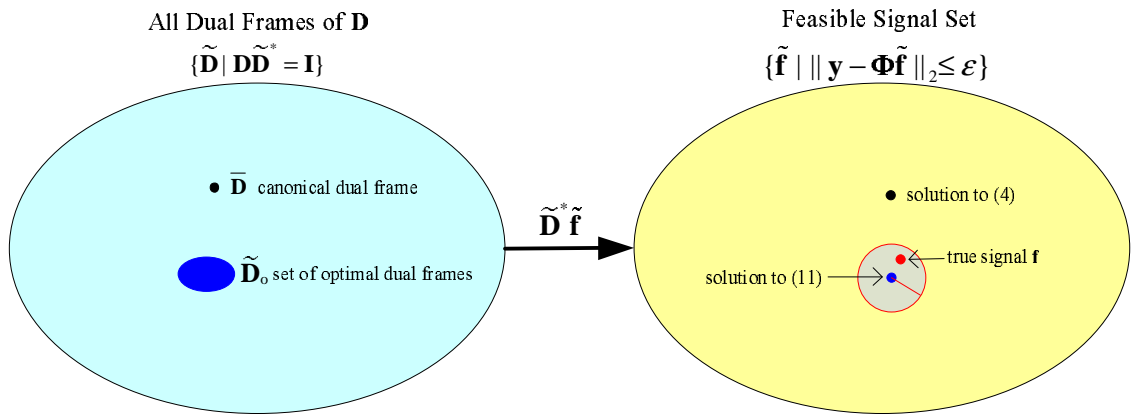}
 \end{tabular}
 \end{center}
 \caption{A schematic
overview of the family of dual-based $\ell_1$-analysis approaches}
 \label{Fig1}
 \end{figure}

%%%%%%%%%%%%%%%%%%%%%%%%%%%%%%%%%
\section{Performance Analysis of $\ell_1$-Synthesis} \label{section3}
%%%%%%%%%%%%%%%%%%%%%%%%%%%%%%%%%
In this section, we present a new performance analysis of the
$\ell_1$-synthesis approach. We begin by showing that the
$\ell_1$-synthesis and the optimal-dual-based $\ell_1$-analysis
approaches are equivalent.

\begin{Theorem}\label{equivalence}
$\ell_1$-synthesis and optimal-dual-based $\ell_1$-analysis are
equivalent.
\end{Theorem}
\begin{proof}
  We start with the optimal-dual-based $\ell_1$-analysis approach as
  posed in \eqref{optimaldualbasedL1analysis2}. Let $\tilde{\xbf} = \bar{\Dbf}^{*}\tilde{\fbf} + \Pbf \gbf
  $, then we have $\Dbf \tilde{\xbf} = \tilde{\fbf}$. Since both $\tilde{\fbf}$ and $\gbf$
  are free, then $\tilde{\xbf} \in \Rnum^d$. Put the two facts into
  \eqref{optimaldualbasedL1analysis2}, we obtain the
  $\ell_1$-synthesis method \eqref{L1synthesis}. On the other hand, we start from the
  $\ell_1$-synthesis formulation. For any $\tilde{\xbf} \in
  \Rnum^d$, the following decomposition always holds
  \begin{align*}
    \tilde{\xbf} & = \tilde{\xbf}_R + \tilde{\xbf}_N = \Dbf^*(\Dbf\Dbf^*)^{-1}\Dbf \tilde{\xbf} + \Pbf\tilde{\xbf}\\
                 & = \bar{\Dbf}^*\Dbf \tilde{\xbf} +
                 \Pbf\tilde{\xbf},
  \end{align*}
  where $\tilde{\xbf}_R$ and $\tilde{\xbf}_N$ are the components of
  $\tilde{\xbf}$ belonging to the row space and the null space of
  $\Dbf$, respectively. Define $\tilde{\fbf} = \Dbf \tilde{\xbf} \in \Rnum^n$ and $\gbf = \tilde{\xbf} \in \Rnum^d$, we can arrive
  at the optimal-dual-based $\ell_1$-analysis approach and the two
  methods are equivalent.
\end{proof}

\noindent{\bf Remark 1:}\ By taking a geometrical description, it
was shown in \cite{Elad2007} that any $\ell_1$-analysis problem
(with full-rank analysis operator) may be reformulated as an
equivalent $\ell_1$-synthesis one. Our results indicate that the
reverse is also true. For a given $\ell_1$-synthesis problem, there
exist some appropriate analysis operators (e.g., optimal dual frames
of $\Dbf$) such that the corresponding $\ell_1$-analysis problem is
equivalent to the $\ell_1$-synthesis one.

With this equivalence, we now establish the error bound of the
$\ell_1$-synthesis approach. Since $\tilde{\Dbf}_{\text{o}}$ is some
alternative dual frame of $\Dbf$, i.e.,
$\Dbf\tilde{\Dbf}_{\text{o}}^* = \Ibf$, a direct application of
Theorem \ref{thm1} leads to the following results.

\begin{Theorem}\label{thm2}
  Let $\Dbf$ be a general frame of $\Rnum^{n}$ with frame bounds $0<A\leq B<\infty$.  Let $\tilde{\Dbf}_{\text{o}}$ be
  some optimal dual frame of $\Dbf$ defined in \eqref{optimalduals} with frame bounds
$0<\tilde{A}_{\text{o}}\leq \tilde{B}_{\text{o}}<\infty$, and let
$\rho=s/b$. Suppose
  \begin{equation}
    \label{SufficientCondtionO} \left(1-\sqrt{\rho B \tilde{B}_{\text{o}}}\right)^2 \cdot \delta_{s+a} +
\rho B \tilde{B}_{\text{o}}\cdot\delta_{b} < 1 - 2\sqrt{\rho B
\tilde{B}_{\text{o}}}
  \end{equation}
holds for some positive integers $a$ and $b$ satisfying $0< b-a\leq
3a$. Then the solution $\hat{\fbf}$ to \eqref{L1synthesis} (or to
\eqref{optimaldualbasedL1analysis2}) satisfies
  \begin{equation}
    \label{ErrorBoundofODBL1analysis} \Vert \hat{\fbf}-\fbf \Vert_{2} \leq C_{0}\cdot\epsilon +
    C_{1}\cdot \f{\|\tilde{\Dbf}_{\text{o}}^*\fbf-(\tilde{\Dbf}_{\text{o}}^*\fbf)_{s}\|_{1}}{\sqrt{s}},
  \end{equation}
 where $C_{0}$ and $C_{1}$ are some constants and $(\tilde{\Dbf}_{\text{o}}^*\fbf)_{s}$
 denotes the vector consisting the largest $s$ entries of
 $\tilde{\Dbf}_{\text{o}}^*\fbf$ in magnitude.
\end{Theorem}

Theorem \ref{thm2} shows that, under suitable conditions on the
sensing matrix, the recovered signal $\hat{\fbf}$ by
$\ell_1$-synthesis is very accurate provided that
$\tilde{\Dbf}_{\text{o}}^*\fbf$ has rapidly decreasing coefficients.
By the optimality of $\tilde{\Dbf}_{\text{o}}$, one may expect that
$\tilde{\Dbf}_{\text{o}}$ will promote high sparsity in the frame
expansion of the signal $\fbf$. Indeed, as we shall see in the
numerical experiments, $\tilde{\Dbf}_{\text{o}}$ is much more
effective in sparsifying $\fbf$ than the canonical dual frame does.
Consequently, in comparison to the standard $\ell_1$-analysis
approach, a better signal recovery is often achieved by
$\ell_1$-synthesis.

More importantly, this new performance analysis result is capable of
explaining examples of successful solutions and fine approximations
by the $\ell_1$-synthesis approach while the recovered coefficient
vector $\xbf$ is no where near its true value.  Known performance
analysis results would not have such capacity.

%%%%%%%%%%%%%%%%%%%%%%%%%%%%%%%%%
\section{Numerical Results} \label{section4}
%%%%%%%%%%%%%%%%%%%%%%%%%%%%%%%%%
In this section, we present some numerical experiments to
demonstrate the effectiveness of the performance analysis results
for $\ell_1$-synthesis. In these experiments, we use two types of
frames: Gabor frames and a concatenation of the coordinate and
Fourier bases. The sensing matrix $\Phibf$ is a Gaussian matrix with
$m=32$, $ n=128$. Since the dependence on the noise in the error
bound \eqref{ErrorBoundofODBL1analysis} is optimal and for the
purpose of clarity, we only consider the noise-free case. Both
$\ell_1$-analysis and $\ell_1$-synthesis problems are solved by the
algorithm developed in \cite{Liu2011} because the returned auxiliary
variable ($\Pbf\gbf$) by this algorithm can be used to construct the
optimal dual frame $\tilde{\Dbf}_{\text{o}}$ \eqref{OptimaldualLS}.
For completeness of this paper, this algorithm is included in
Appendix \ref{Appendix1}. We set $\lambda=\mu=1$, $tol=10^{-12}$,
$nInner = 5$, and $nOuter = 100$ in this algorithm  for all
experiments.

\noindent {\bf Example 1: Gabor Frames.}\ \  Recall that for a
window function $g$ and positive time-frequency shift parameters
$\alpha$ and $\beta$, the Gabor frame is given by
\begin{equation}
  \{g_{_{l,k}}(t) = g(t-k\alpha)e^{2\pi il\beta t}\}_{l,k}.
\end{equation}
For many practical applications such as radar and sonar, the
received signal $f$ often has the form
\begin{equation}
  f(t) = \sum_{k=1}^s a_k g(t-t_k)e^{i\omega_k t}.
\end{equation}
Evidently, $f$ is sparse with respect to a Gabor frame. In this
experiment, we construct a Gabor dictionary with Gaussian windows,
oversampled by a factor of $30$ so that $d=30\times n=3840$. The
tested signal $\fbf$ is sparse with respect to the constructed Gabor
frame with sparsity $s = \text{ceil}(0.2\times m) = 7$. The
positions of the nonzero entries of the coefficient vector $\xbf$
are selected uniformly at random, and each nonzero value is sampled
from standard Gaussian distribution.

Figure \ref{Fig2} (a) shows that when $\Dbf$ is highly coherent with
coherence\footnote{The coherence of the dictionary $\Dbf$ is defined
as $\mu(\Dbf) = \underset {j \neq k} {\textrm{max}}\frac{|\langle
\dbf_j, \dbf_k\rangle|}{\|\dbf_j\|_2\|\dbf_k\|_2}$, where $\dbf_j$
and $\dbf_k$ denote columns of $\Dbf$. We say that $\Dbf$ is
incoherent if $\mu(\Dbf)$ is small.} $\mu(\Dbf)=0.9934$, the
recovered coefficients by the $\ell_1$-synthesis are disappointing
(with a relative error $\|\bar{\xbf}-\xbf\|_2/\|\xbf\|_2=0.9039$).
However, the signal recovered by the $\ell_1$-synthesis is
nevertheless quite acceptable (with a relative error equal to
$0.0845$), see Figure \ref{Fig2} (b). This example tells us that a
good recovery of the coefficients $\xbf$ may be unnecessary to
guarantee a fine reconstruction of the signal $\fbf$.  This
phenomenon is explainable by the new performance analysis result,
but not by performance results based on the accuracy of the recovery
of the coefficient vector $\xbf$.

Figure \ref{Fig2} (b) also shows that the signal recovery via
$\ell_1$-synthesis is much better than that of $\ell_1$-analysis
(relative error: $0.0845$ vs. $0.3445$). This is because the optimal
dual frame $\tilde{\Dbf}_{\text{o}}$ is much more effective in
promoting sparsity in the frame expansion of $\fbf$ than the
canonical dual frame $\bar{\Dbf}$ does. Figure \ref{Fig2} (c)
compares the largest 100 coefficients (in magnitude) of
$\bar{\Dbf}^*\fbf$ and $\tilde{\Dbf}_{\text{o}}^*\fbf$, where
$\tilde{\Dbf}_{\text{o}}^*$ is determined by \eqref{OptimaldualLS}.

\noindent {\bf Example 2: Concatenations.}\ \ When signals of
interest are sparse over several orthonormal bases (or frames), it
is natural to use a dictionary $\Dbf$ consisting of a concatenation
of these bases (or frames). In this experiment, we consider a
dictionary consisting of the coordinate and Fourier bases, i.e.,
$\Dbf = [\Ibf, \Fbf]/\sqrt{2}$. The tested signal $\fbf$ is a linear
combination of spikes and sinusoids, i.e., $\fbf = \fbf_1 + \fbf_2 =
\xbf_1 + \Fbf\xbf_2$. The sparsity of both $\xbf_1$ and $\xbf_2$ is
equal to $4$. Again, the positions of the nonzero entries of both
$\xbf_1$ and $\xbf_2$ are selected uniformly at random, and each
nonzero value is sampled from standard Gaussian distribution.

Figure \ref{Fig3} (a) and (b) show that when $\Dbf$ is incoherent
(with coherence $\mu(\Dbf)=1/\sqrt{n}=0.0884$), the
$\ell_1$-synthesis approach not only recovers the signal $\fbf$ but
also the coefficient vector $\xbf$ accurately.

Figure \ref{Fig3} (b) also shows that $\ell_1$-analysis fails in
recovering the signal with a relative error at $0.8143$. Such a
failure is not surprising since $\bar{\Dbf}=\Dbf$ is ineffective in
sparsifying the true signal $\fbf$, see Figure \ref{Fig3} (c). By
contrast, $\tilde{\Dbf}_{\text{o}}^*\fbf$ decays very quickly, which
guarantees the good recovery for the signal by $\ell_1$-synthesis.

 \begin{figure}
 \begin{center}
 \begin{tabular}{c}
 \includegraphics[height=6.5cm]
 {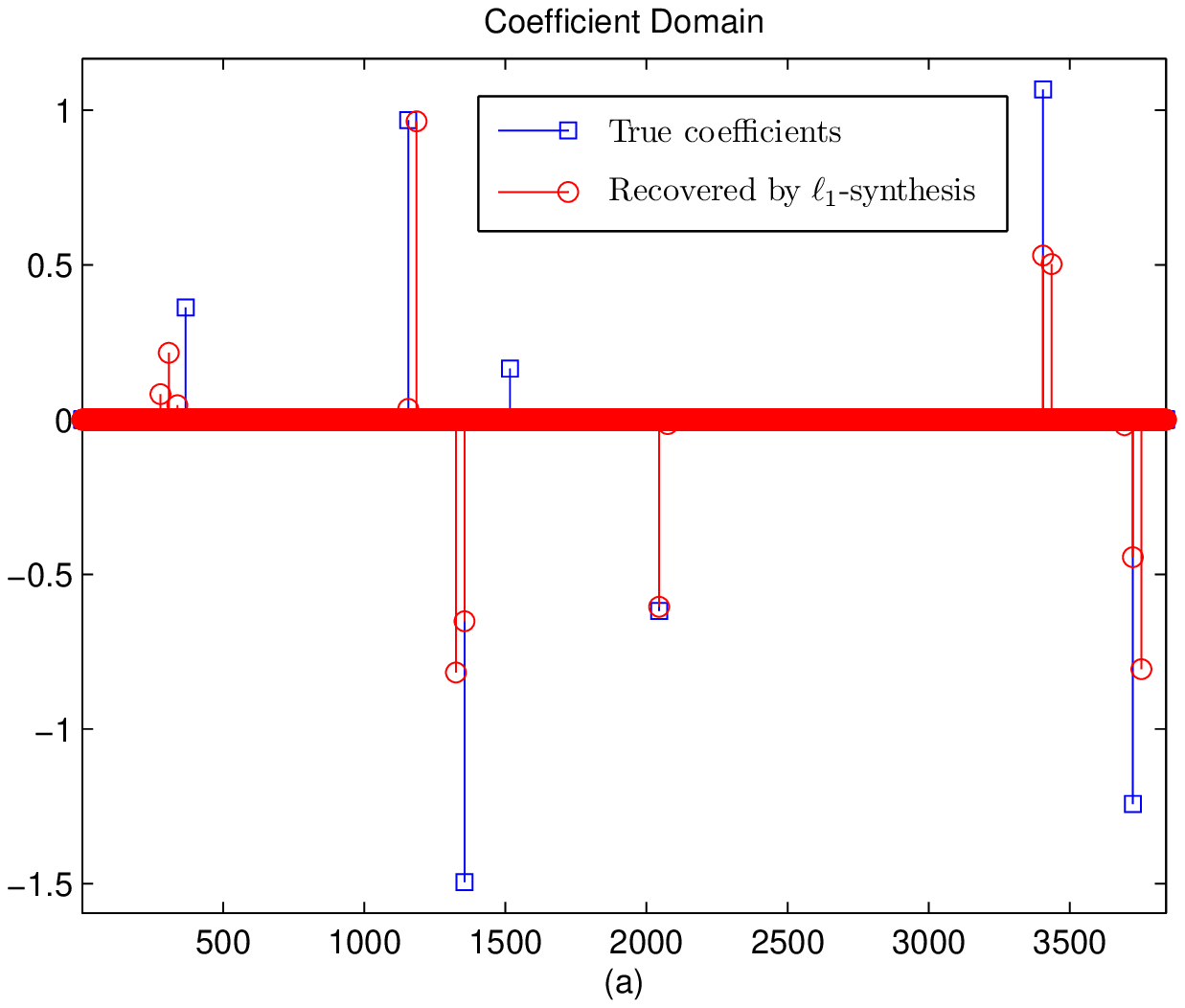}\\
 \includegraphics[height=6.5cm]
 {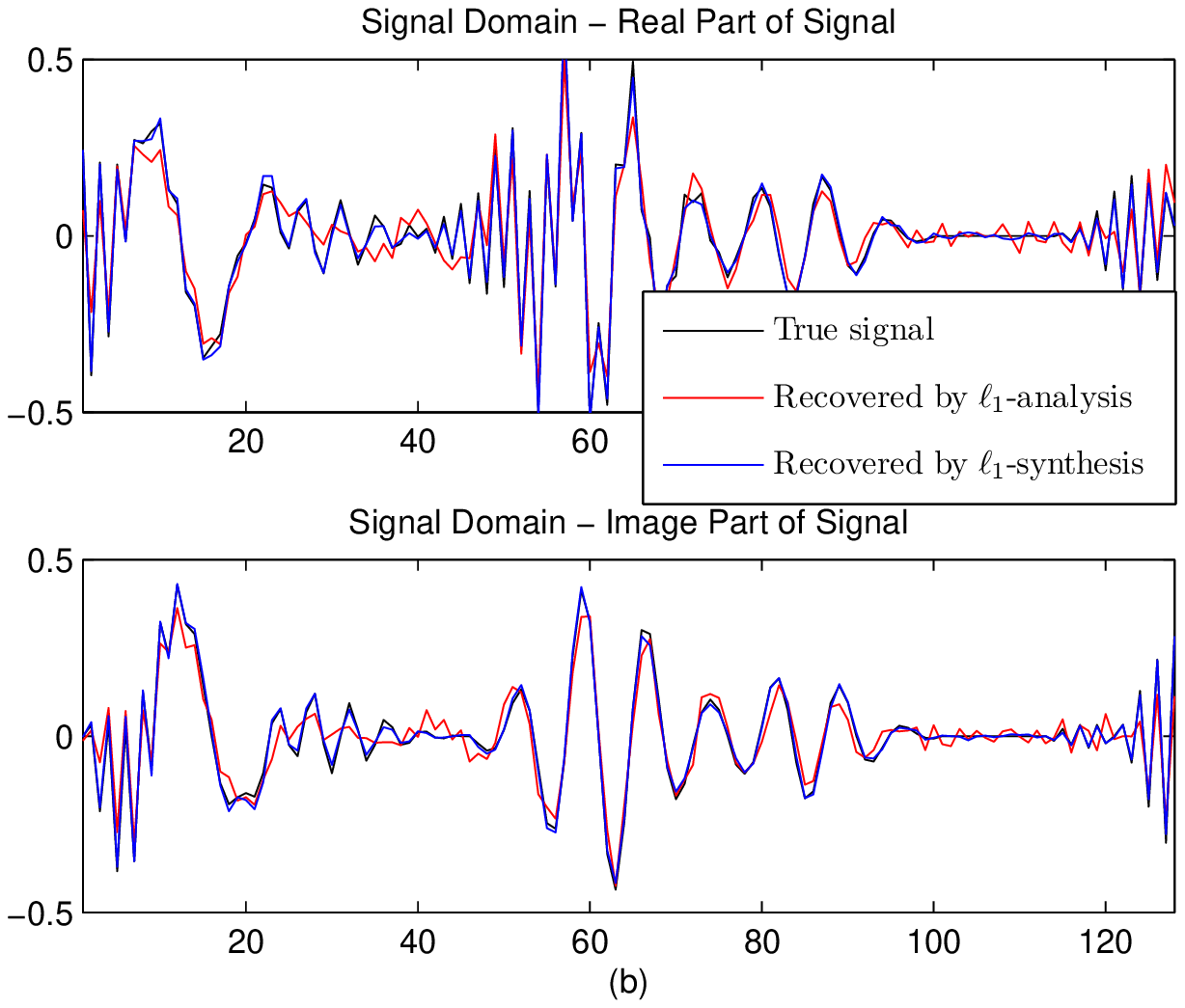}\\
 \includegraphics[height=6.5cm]
 {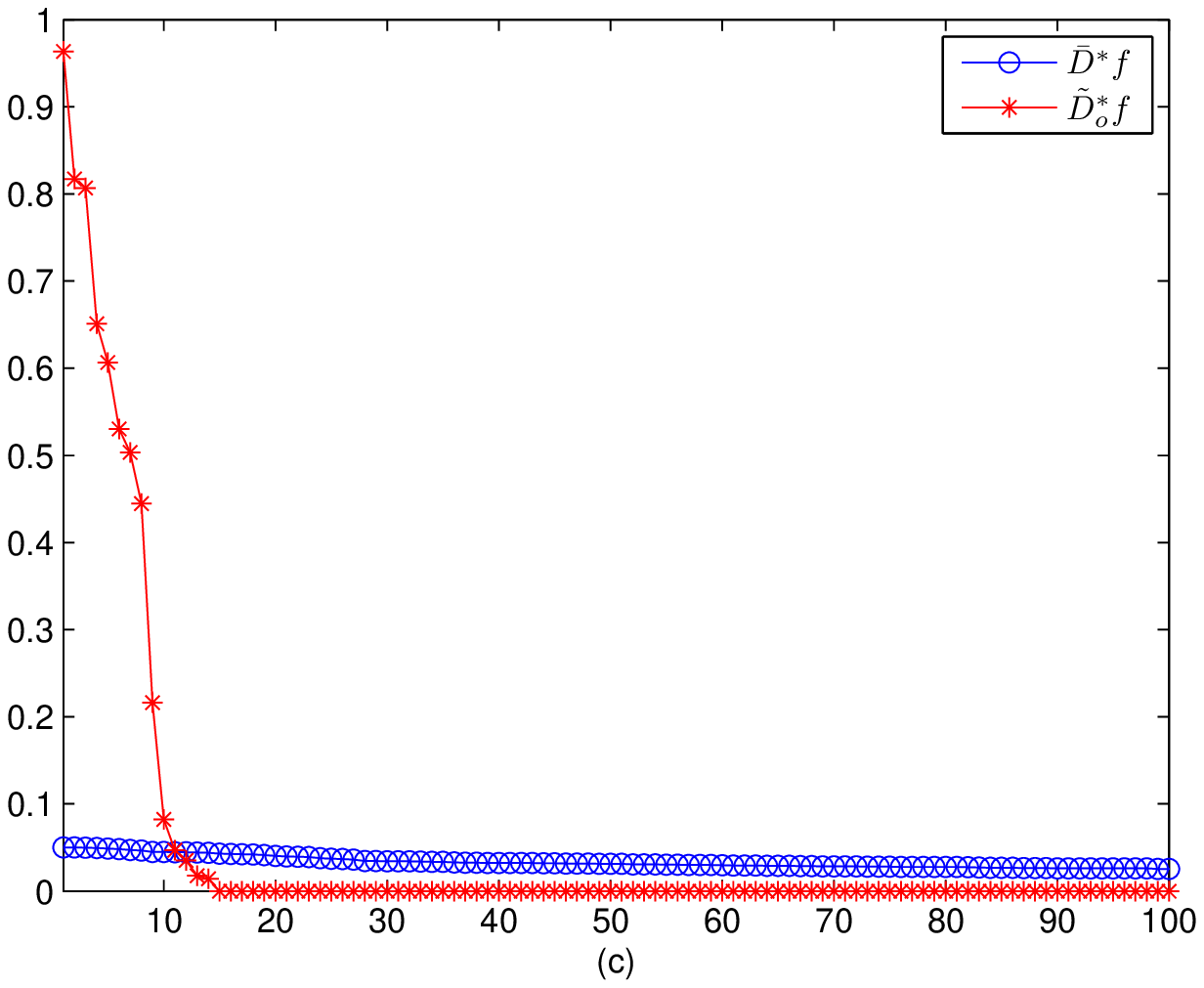}
 \end{tabular}
 \end{center}
 \caption{$\Dbf = $ Gabor frame. (a): recovery in coefficient domain by $\ell_1$-synthesis (relative error: $0.9039$) with the relative error defined as $\|\hat{\xbf}-\xbf\|_2/\|\xbf\|_2$.
 (b): recovery in signal domain by $\ell_1$-analysis (relative error: $0.3445$)
 and $\ell_1$-synthesis (relative error: $0.0845$) with the relative error defined as $\|\hat{\fbf}-\fbf\|_2/\|\fbf\|_2$.
 (c): The largest $100$ coefficients of the coefficient vector $\bar{\Dbf}^*\fbf$ and $\tilde{\Dbf}_{\text{o}}^*\fbf$ in magnitude.}
 \label{Fig2}
 \end{figure}

 \begin{figure}
 \begin{center}
 \begin{tabular}{c}
 \includegraphics[height=6.5cm]
 {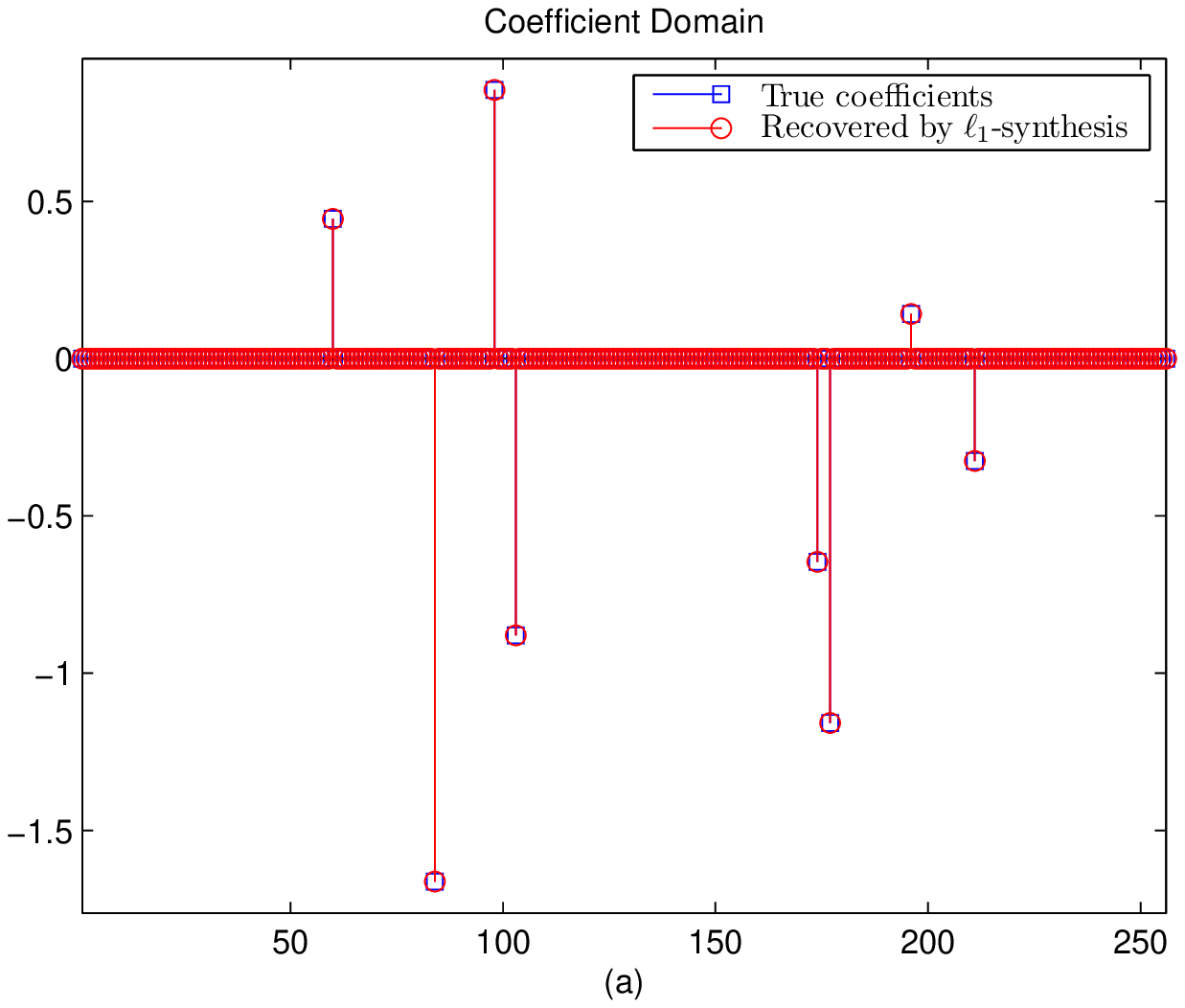}\\
 \includegraphics[height=6.5cm]
 {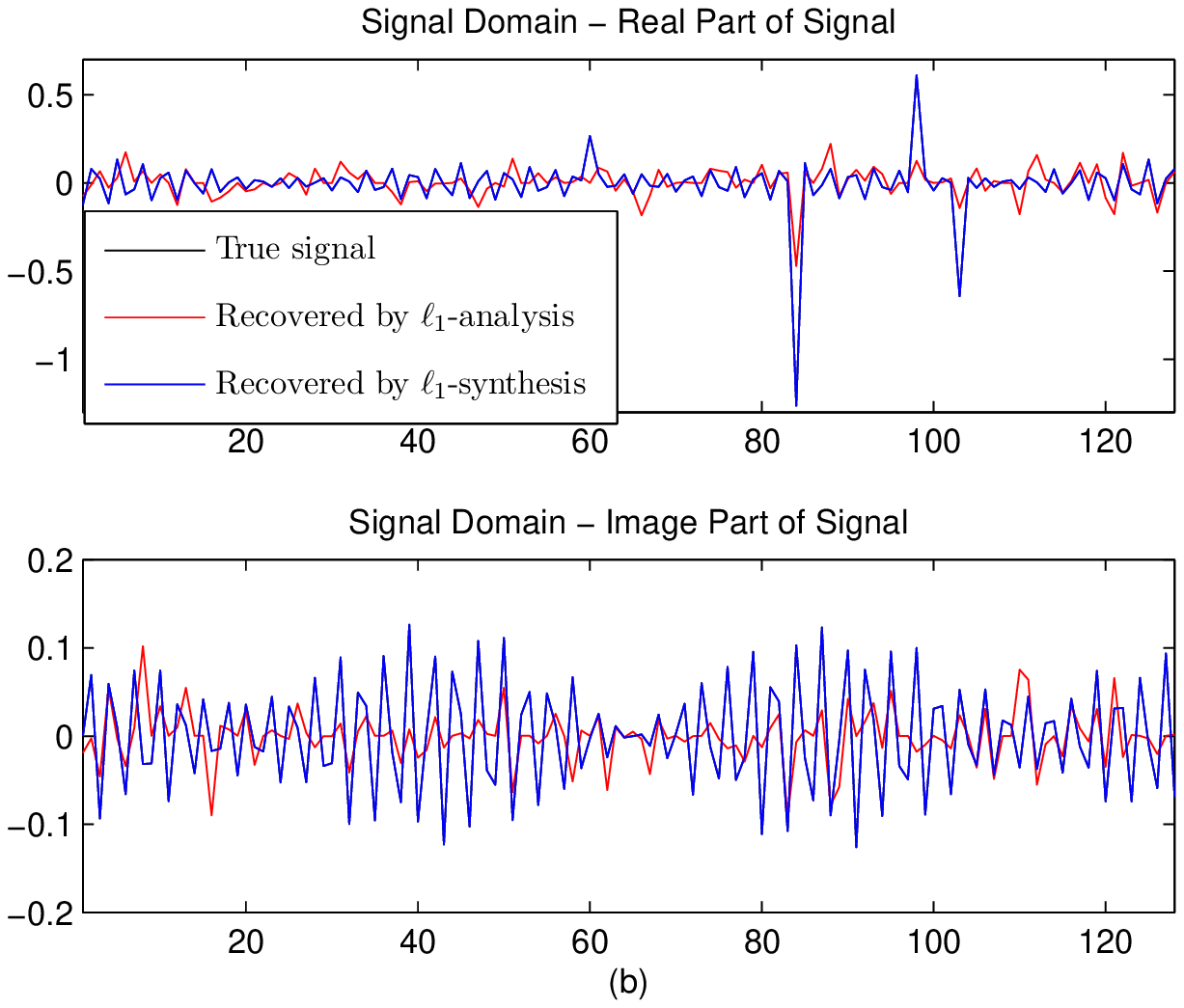}\\
 \includegraphics[height=6.5cm]
 {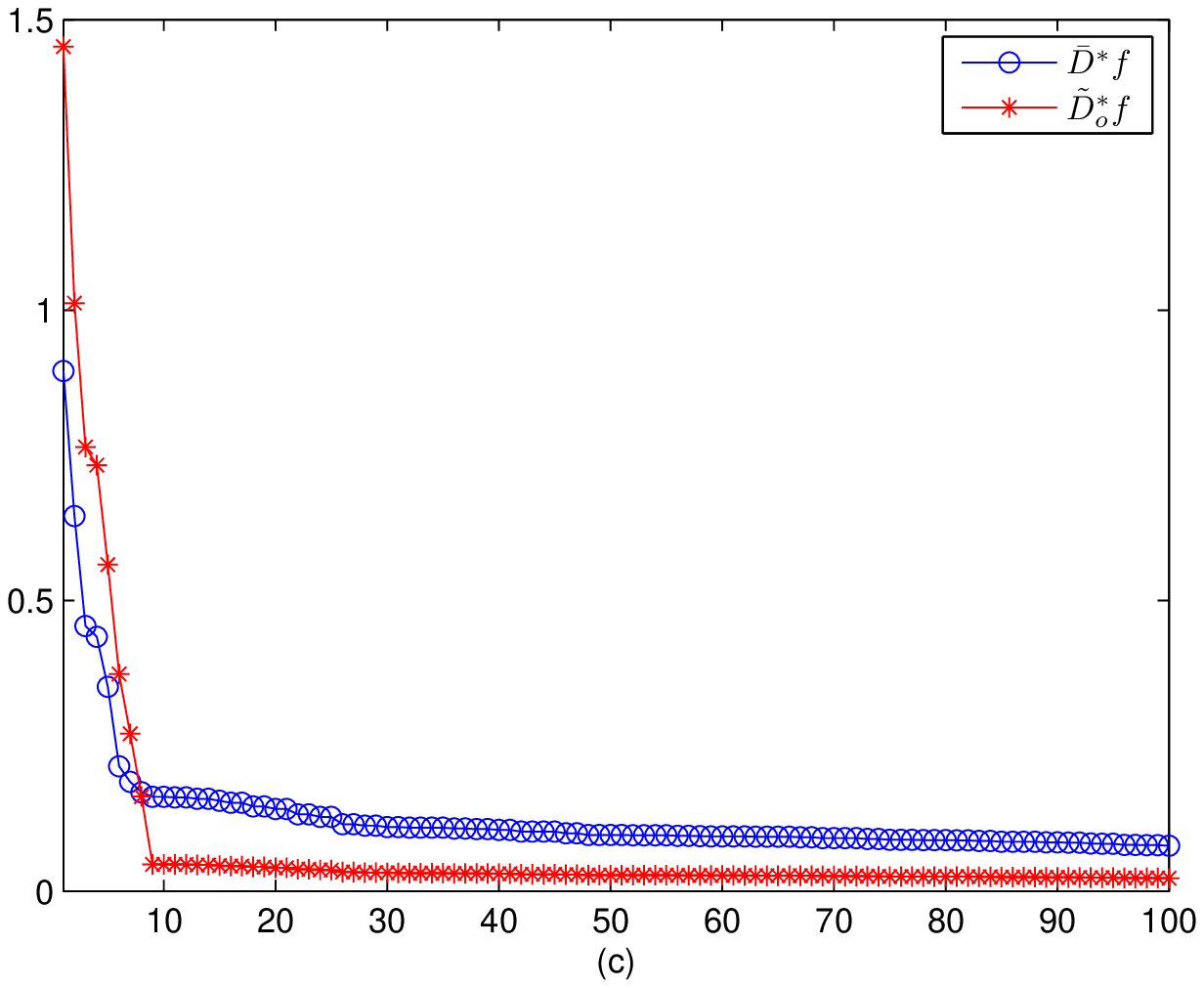}
 \end{tabular}
 \end{center}
 \caption{$\Dbf = [\Ibf, \Fbf]/\sqrt{2}$. (a): recovery in coefficient domain by $\ell_1$-synthesis (relative error: $7.7681\times 10^{-6}$) with the relative error defined as $\|\hat{\xbf}-\xbf\|_2/\|\xbf\|_2$.
 (b): recovery in signal domain by $\ell_1$-analysis (relative error: $0.8143$)
 and $\ell_1$-synthesis (relative error: $6.7911\times 10^{-6}$) with the relative error defined as $\|\hat{\fbf}-\fbf\|_2/\|\fbf\|_2$.
 (c): The largest $100$ coefficients of the coefficient vector $\bar{\Dbf}^*\fbf$ and $\tilde{\Dbf}_{\text{o}}^*\fbf$ in magnitude.}
 \label{Fig3}
 \end{figure}

%%%%%%%%%%%%%%%%%%%%%%%%%%%%%%%%%
\section{Conclusions} \label{section5}
%%%%%%%%%%%%%%%%%%%%%%%%%%%%%%%%%
This paper has presented a novel performance analysis for
$\ell_1$-synthesis in which the dictionary may be highly coherent.
Our approach was to show the equivalence between $\ell_1$-synthesis
and optimal-dual-based $\ell_1$-analysis. With this equivalence, the
signal recovery error bound for both could be established by using
the results in \cite{Liu2011}. Finally, the results obtained in this
paper were validated via numerical experiments.

\appendix
\section{Split Bregman Iteration for optimal-dual-based
$\ell_1$-analysis} \label{Appendix1} This appendix includes the
split Bregman iteration for optimal-dual-based $\ell_1$-analysis in
which
\begin{itemize}
  \item $\fbf$: the recovered signal;
  \item $\xbf$: the recovered coefficient vector;
  \item $\Pbf\gbf$: the auxiliary
   variable used to construct the optimal dual frame of $\Dbf$;
  \item shrink($\cdot$): denotes the element-wise soft shrinkage
  operation;
  \item  $(\cdot)^{new}$: denotes either $(\cdot)^{k+1}$ if it is
  available or $(\cdot)^{k}$ otherwise.

\end{itemize}

\begin{algorithm} \label{algorithm1}
\caption{Split Bregman Iteration for optimal-dual-based
$\ell_1$-analysis}
  \textbf{Initialization:} {$\fbf^0=\zebf$, $\xbf^0=\bbf^0=\Pbf\gbf^0=\zebf$, $\cbf^0=\zebf$, $\mu>0, \lambda>0, nOuter, nInner,
  tol$}\;
  \textbf{Output:} {$\fbf$, $\xbf$, $\Pbf\gbf$}\;
  \While{$k < nOuter$ and $\|\Phibf \fbf^k - \ybf\|_2 > tol$}
  { \For {$n=1:nInner$}{
  $\fbf^{k+1} = (\mu\Phibf^*\Phibf + \lambda
  \bar{\Dbf}\bar{\Dbf}^{*})^{-1}[\mu \Phibf^*(\ybf -
  \cbf^k) + \lambda \bar{\Dbf} (\xbf^{new} - \Pbf \gbf^{new} -
  \bbf^{new})]$\;
  $\xbf^{k+1} = \text{shrink} (\bar{\Dbf}^{*}\fbf^{new}+\Pbf \gbf^{new} +
  \bbf^{new}, 1/\lambda)$\;
  $\Pbf\gbf^{k+1} = \Pbf(\xbf^{new} - \bar{\Dbf}^{*}\fbf^{new} -
  \bbf^{new})$\;
  $\bbf^{k+1}  = \bbf^{new} + (\bar{\Dbf}^{*}\fbf^{new} + \Pbf \gbf^{new} -
  \xbf^{new})$\;
  }
  $\cbf^{k+1} = \cbf^k + (\Phibf\fbf^{k+1} - \ybf)$\;
  Increase $k$\;
  }
\end{algorithm}

% 一个插入单个图的例子
% figure 1
% \begin{figure}
% \begin{center}
% \begin{tabular}{c}
% \includegraphics[height=8cm]
% {Name of Figure}
% \end{tabular}
% \end{center}
% \caption{Description for the Figure}
% \label{Fig1}
% \end{figure}

% 一个插入多个图的例子
% % figure 2
% \begin{figure}
% \begin{center}
% \begin{tabular}{cc}
% \includegraphics[height=5cm]
% {Name of Figure}&
% \includegraphics[height=5cm]
% {Name of Figure}\\
% (a) &  (b) \\
% \end{tabular}
% \end{center}
% \caption{Description of Figures}
% \label{Fig2}
% \end{figure}

%%一个插入表格的例子
%\begin{table}[htbp]
%%\vspace{0cm}
%\caption{Name of Table}
%\begin{center}
%\begin{tabular}{cccc}
%  \hline\hline
%   Item1 & Item2  & Item3  & Item4  \\
%  \hline\hline  %Excitation
%  20 21  & 0.252  & 0.828  & 10 31  \\
%  19 22  & 0.823  & 0.958  & 9 32   \\
%  18 23  & 1.429  & 0.944  & 8 33   \\
%  17 24  & 2.064  & 0.987  & 7 34   \\
%  \hline\hline
%\end{tabular}
%\end{center}
%%\vspace{-.5cm}
%\end{table}

\spacingset{1}
%\newpage
\bibliographystyle{ieeetr}
\bibliography{Sparse Recovery with Coherent and Redundant Dictionaries.bbl}
%\newpage
%\input{fig}

%\begin{thebibliography}{1}
%
%\bibitem{Haykin2002}
%S.~Haykin, {\em Adaptive Filter Theoty.} 4th ed. \newblock Upper
%Saddle River, NJ: Prentice Hall, 2002.
%
%\end{thebibliography}

\end{document}